\documentclass{article}
\usepackage[english]{babel}
\usepackage[utf8]{inputenc}
\usepackage{amsmath,amsthm, amssymb}
\usepackage{amsfonts, amsbsy}
\usepackage[mathscr]{euscript}
\usepackage{xspace}
\usepackage{graphicx,subfigure}
\usepackage{color}
\usepackage{enumerate}

\theoremstyle{plain}
\newtheorem{theorem}{Theorem}
\newtheorem{corollary}{Corollary}
\newtheorem{lemma}{Lemma}

\title{On the Existence of Tree Backbones that Realize the Chromatic Number on a Backbone Coloring}

\author{J.~Araujo and A. A.~Cezar and A.~Silva\\
ParGO Group - Parallellism, Graphs and Optimization\\Departamento de Matem\'atica\\Universidade Federal do Cear\'a, Fortaleza, Brazil
}
\begin{document}

\maketitle

\begin{abstract}
A proper $k$-coloring of a graph $G=(V,E)$ is a function $c: V(G)\to \{1,\ldots,k\}$ such that $c(u)\neq c(v)$, for every $uv\in E(G)$. The chromatic number $\chi(G)$ is the minimum $k$ such that there exists a proper $k$-coloring of $G$.
Given a spanning subgraph $H$ of $G$, a $q$-backbone $k$-coloring of $(G,H)$ is a proper $k$-coloring $c$ of $V(G)$ such that $\lvert c(u)-c(v)\rvert \ge q$, for every edge $uv\in E(H)$. The $q$-backbone chromatic number $BBC_q(G,H)$ is the smallest $k$ for which there exists a $q$-backbone $k$-coloring of $(G,H)$. In this work, we show that every connected graph $G$ has a spanning tree $T$ such that $BBC_q(G,T) = \max\{\chi(G),\left\lceil\frac{\chi(G)}{2}\right\rceil+q\}$, and that this value is the best possible.

As a direct consequence, we get that every connected graph $G$ has a spanning tree $T$ for which $BBC_2(G,T)=\chi(G)$, if $\chi(G)\ge 4$, or $BBC_2(G,T)=\chi(G)+1$, otherwise. Thus, by applying the Four Color Theorem, we have that every connected nonbipartite planar graph $G$ has a spanning tree $T$ such that $BBC_2(G,T)=4$. This settles a question by Wang, Bu, Montassier and Raspaud (2012), and generalizes a number of previous partial results to their question.
\end{abstract}

%\begin{keyword}
%	Backbone coloring\sep Backbone chromatic number\sep Tree backbone\sep Planar graphs.
%\end{keyword}

%\end{frontmatter}

\section{Introduction}

For basic notions and terminology on Graph Theory, the reader is referred to~\cite{BM08}. All graphs in this work are considered to be simple. Because we investigate the existence of a spanning tree with certain property, we also consider only connected graphs. However, for disconnected graphs, the statements hold by replacing ``spanning tree'' by ``spanning forest''. A \emph{proper $k$-coloring} of a graph $G$ is a function $c:V(G)\rightarrow \{1,\ldots,k\}$ such that $c(u)\neq c(v)$, for every $uv\in E(G)$. If $G$ admits a proper $k$-coloring, we say that $G$ is \emph{$k$-colorable}. The \emph{chromatic number} of $G$, denoted by $\chi(G)$, is the smallest positive integer $k$ such that $G$ is $k$-colorable. Determining the chromatic number of a graph is an NP-hard problem on Karp's list~\cite{Karp72} and one of the most studied problems on Graph Theory~\cite{JT95, MR01}.

Given a spanning subgraph $H$ of $G$, and positive integers $k$ and $q$, a \emph{$q$-backbone $k$-coloring} of $(G,H)$ is a proper $k$-coloring $c$ of $G$ such that $\lvert c(u)-c(v)\rvert\ge q$, for every $uv\in E(H)$. The \emph{$q$-backbone chromatic number of $(G,H)$}, denoted by $BBC_q(G,H)$, is the smallest integer $k$ for which $(G,H)$ admits a $q$-backbone $k$-coloring. 

This parameter was first introduced by Broersma et al. \cite{BFGW03} as a model for the frequency assignment problem where certain channels of communication are more demanding than others. In their seminal work, they only considered $q=2$ and they were interested in finding out how far away from $\chi(G)$ can $BBC_2(G,H)$ be in the worst case. Concerning trees, for each positive integer $k$, they defined:

$$\mathcal{T}_k=\max\{BBC_2(G,T): \mbox{$\chi(G)=k$ and $T$ is a spanning tree of $G$}\}.$$

Note that, if $c$ is a proper $\chi(G)$-coloring of $G$, then by recoloring each vertex $u$ with color $2c(u)-1$, we obtain a proper $(2\chi(G)-1)$-coloring of $G$ where every color is odd. Therefore, we get $BBC_2(G,G)\le 2\chi(G)-1$. This gives an upper bound of $2k-1$ for $\mathcal{T}_k$. In \cite{BFGW03}, they proved that this is actually best possible.

\begin{theorem}[Broersma et al.\cite{BFGW03}]
$\mathcal{T}_k = 2k-1$, for every positive integer $k$.
\end{theorem}

This means that, between all the $k$-colorable graphs, there is one that attains this upper bound. However, it does not give any insight on how bad can a tree backbone be for a given graph $G$. One could then define $\mathcal{T}_2(G)$ as the maximum $BBC_2(G,T)$, where $T$ is a spanning tree of $G$.
This worst case behaviour has been studied for planar graphs. If $G$ is planar, because $\chi(G)\le 4$ and the fact that $BBC_2(G,G)\le 2\chi(G)-1$, we get $\mathcal{T}_2(G)\le 7$. Broersma et al. \cite{BFGW07} give examples where $BBC_2(G,T)=6$, and conjecture that $\mathcal{T}_2(G)=6$. A partial result for their conjecture has been given in~\cite{CHSS13}. Note that this parameter can be generalized for higher values of $q$. In~\cite{HKLT14}, Havet et al. prove that, if $G$ is a planar graph, then $\mathcal{T}_q(G)\le q+6$. They also prove that this is best possible if $q\ge 4$, and conjecture that $\mathcal{T}_3(G)\le 8$.

Now, observe that it is not clear whether $G$ always has a spanning tree with a ``good'' behaviour, i.e., such that $BBC_q(G,T)$ is not much larger than $\chi(G)$. Therefore, it makes sense to define the best case behaviour of $BBC_q(G,T)$. In \cite{WBMR12}, Wang, Bu, Montassier and Raspaud asked what is the smallest value $\beta$ for which the following holds: if $G$ is a nonbipartite planar graph with girth at least $\beta$, then $G$ has a spanning tree $T$ such that $BBC_2(G,T) = 4$. Inspired by their question, we define the following parameter, for a given graph $G$ and a positive integer $q$:

$\mathcal{B}_q(G) = \min\{BBC_q(G, T): \mbox{$T$ is a spanning tree of $G$}\}.$

Our main result is the following:

\begin{theorem}\label{thm:main}
For every graph $G$ and positive integer $q$, \[\mathcal{B}_q(G) = \max\{\chi(G),\left\lceil \frac{\chi(G)}{2}\right\rceil +q \}.\]
\end{theorem}

This gives us the following value for bipartite graphs:

\begin{corollary}
If $G$ is bipartite, then $\mathcal{B}_q(G)=q+1$.
\end{corollary}

Considering $q\ge 2$, observe that if $G$ has at least one edge and $T$ is a spanning tree of $G$, then $BBC_2(G,T)\geq 3$, and that $BBC_2(G,T) = 3$ if, and only if, $G$ is bipartite. Also, observe that, when $G$ is a nonbipartite planar graph, we get that $\max\{\chi(G),\lceil \chi(G)/2\rceil + 2\}$ is always equal to 4. Therefore, the answer to Wang et al's question is $\beta = 3$, i.e., having high girth is not a necessary condition for having the desired spanning tree. 

\begin{corollary}
If $G$ is a nonbipartite planar graph and $q\ge 2$, then $\mathcal{B}_q(G)=q+2$. In particular, $G$ always has a spanning tree $T$ for which $BBC_2(G,T)=4$.
\end{corollary}

We mention that this generalizes results in a number of papers:~\cite{BB15,BL11,BZ10,BZ11,WBMR12}. We also mention that, in \cite{WBMR12}, Wang et al. wrongly state that $\beta$ is at least 4 due to the existence of a nonbipartite planar  graph $G$ and a spanning tree $T$ of $G$ such that $BBC_2(G,T)=6$. However, they fail to notice that, in order for $\beta$ to be at least 4, this should hold for every spanning tree of $G$.

%%%%%%%%%%%%%%%%%%%%%%%%%%%%%%%%%%%%%%%%%%%%%%%%%%%%%%%%%%%%%%%%%%%%%%%%
%%%%%%%%%%%%%%%%%%%%%%%%%%%%%%%%%%%%%%%%%%%%%%%%%%%%%%%%%%%%%%%%%%%%%%%%
%%%%%%%%%%%%%%%%%%%%%%%%%%%%%%%%%%%%%%%%%%%%%%%%%%%%%%%%%%%%%%%%%%%%%%%%

\section{Proof of Theorem~\ref{thm:main}}

Roughly, the idea of the proof is to show that any graph $G$ has a \emph{nice} proper $k$-coloring, where $k = \max\{\chi(G),\left\lceil \chi(G)/2\right\rceil + q\}$. By nice we mean that the subgraph of $G$ induced by the edges whose endpoint colors differ by at least $q$ form a \emph{connected spanning} subgraph of $G$. Then, we select among these edges a spanning tree to form its backbone. Before presenting the main result, let us recall some definitions, and present some new ones.

Consider a proper $k$-coloring $c$ of a graph $G$. For $i\in \{1,\ldots,k\}$, the \emph{color class $i$ of $c$} is the subset $c_i=\{u\in V(G): c(u)=i\}$. Observe that if $H$ is a component of $G[c_i\cup c_j]$, a.k.a. Kempe's chain, then the $k$-coloring $c'$ obtained from $c$ by switching colors $i$ and $j$ in $V(H)$ is also a proper $k$-coloring of $G$. We denote the set of edges $\{uv\in E(G): u\in V(H)\mbox{ and }v\in V(G)\setminus V(H)\}$ by $[H,\overline{H}]$. Given an integer $q$, and $i\in \{1,\ldots,k\}$, we denote by $[i]_q$ the set $\{j\in \{1,\ldots,k\}: \lvert i-j\rvert <q\}$. The \emph{$q$-subgraph of $c$}, denoted by $G_{c,q}$, is the subgraph $(V(G),E_{c,q})$, where $E_{c,q} = \{uv\in E(G): \lvert c(u)-c(v)\rvert\ge q\}$. Alternatively, one can see that $uv\in E_{c,q}$ if and only if $c(u)\notin [c(v)]_q$ if and only if $c(v)\notin [c(u)]_q$. Our upper bound is obtained as a corollary of the following theorem:

%$[c(u)]_q\cap[c(v)]_q =\emptyset$ if and only if $c(u)\notin [c(v)]_q$. Our upper bound is obtained as a corollary of the following theorem:

\begin{theorem}
If $G$ is a connected graph and $k\ge \max\{\chi(G), \lceil \chi(G)/2\rceil + q\}$, then there exists a proper $k$-coloring $c$ of $G$ such that $G_{c,q}$ is connected.
\end{theorem}
\begin{proof}
Consider $k=\max\{\chi(G),\lceil \chi(G)/2\rceil +q\}$ and let $c$ be a proper $k$-coloring of $G$ that uses the following $\chi(G)$ colors: $\{1,\ldots,x,x+k'+1,\ldots,k\}$, where $x=\lceil \chi(G)/2\rceil$ and $k'=k-\chi(G)$. Let $H$ be a component of $G_{c,q}$ with maximum number of vertices. Suppose, without loss of generality, that $c$ maximizes the size of $H$. We claim that such a coloring $c$ satisfies that $G_{c,q}$ is connected, which means that $H$ is a spanning subgraph of $G$.

By contradiction, suppose that $V(H)\subset V(G)$, i.e., $H$ does not contain every vertex of $G$. Since $G$ is connected, there must be an edge $uv\in [H,\overline{H}]$. By the definition of $G_{c,q}$, we know that $[c(u)]_q\cap[c(v)]_q\neq \emptyset$.

First, suppose that there exists $j\in \{1,\ldots,k\}\setminus([c(u)]_q\cup [c(v)]_q)$, and let $H'$ be the component of $G[c_j \cup c_{c(v)}]$ containing $v$. We claim that $V(H')\cap V(H)=\emptyset$. Suppose otherwise and let $v'\in V(H')\cap V(H)$ be closest to $v$ in $H'$; also, let $w\in N_{H'}(v')\setminus V(H)$ (it exists by the choice of $v'$). By the definition of $H'$, we know that $\{c(v'),c(w)\} = \{j,c(v)\}$. This contradicts the construction of $H$ since $wv'\notin E_{c,q}$ and $j\notin [c(v)]_q$. Now, let $c'$ be obtained from $c$ by switching colors $j$ and $c(v)$ in $H'$. Because $V(H')\cap V(H)=\emptyset$, nothing changes in $H$; additionally, $c'(v)\notin [c'(u)]_q$, which means that $uv\in E_{c',q}$ and that there is a component in $G_{c',q}$ that strictly contains $H$, a contradiction to the choice of $c$.

Now, suppose that

%\begin{itemize}
%\item[(*)] $[c(u)]_q\cup [c(v)]_q=\{1,\ldots,k\}$, for all $uv\in [H,\overline{H}]$.
%\end{itemize}

\begin{equation}
\label{eq:1}
[c(u)]_q\cup [c(v)]_q=\{1,\ldots,k\}, \text{ for all } uv\in [H,\overline{H}].
\end{equation} 
Recall that $c$ uses the colors that are in the set $\{1,\ldots,x,x+k'+1,\ldots,k\}$, where $k=\max\{\chi(G),\lceil \chi(G)/2\rceil +q\}$, $x=\lceil \chi(G)/2\rceil$ and $k'=k-\chi(G)$. We want to prove that $1\notin [i]_q$, for every $i\in\{x+k'+1,\ldots,k\}$, and that $k\notin [i]_q$, for every $i\in\{1,\ldots,x\}$. We analyse the cases below.

\begin{itemize}
 \item $q\ge \lfloor \chi(G)/2\rfloor$: in this case, $k=x+q$. If $i\in\{1,\ldots,x\}$, then $k-i\ge k-x = x+q-x = q$. In case, $i\in\{x+k'+1,\ldots,k\}$, then $i-1 \ge x+k'+1-1 = x+k-\chi(G) = x+x+q-\chi(G)\ge q$.

 \item $q<\lfloor \chi(G)/2\rfloor$: observe that $k=\chi(G)$ and $k'=0$. If $i\in\{1,\ldots,x\}$, then $k-i\ge k-x = 
 \chi(G)-x = \lfloor \chi(G)/2\rfloor > q$. Similarly, if $i\in\{x+k'+1,\ldots,k\}$, then $i-1 \ge x+k'+1-1 = x> q$.
\end{itemize} 
Now, consider any edge $uv\in [H,\overline{H}]$. Suppose that $c(u)\le x$, in which case $k\notin [c(u)]_q$; if this is not the case, we get $1\notin [c(u)]_q$ and the argument is analogous. By Equation~\ref{eq:1}, we get $k\in [c(v)]_q$, and therefore $c(v)\ge x+k'+1$. Let $H'$ be the component of $G[c_k\cup c_{c(v)}]$ containing $v$. We claim that $V(H')\cap V(H)=\emptyset$. Suppose otherwise, and let $v'\in V(H')\cap V(H)$ be the closest to $v$ in $H'$ and let $w\in N_{H'}(v')\setminus V(H)$. By the choice of $H'$, we know that $\{c(v'),c(w)\} = \{c(v),k\}$, in which case $1\notin [c(v')]_q\cup [c(w)]_q$, contradicting Equation~\ref{eq:1}. Finally, the theorem follows by the same argument used on the previous case.
\end{proof}

It remains to prove that this is also a lower bound. Our proof actually holds for any spanning backbone that does not contain isolated vertices.

\begin{lemma}
If $G$ is a graph and $H$ is a spanning subgraph of $G$ such that $\delta(H)\ge 1$, then, for every positive integer $q$ the following holds:
\[BBC_q(G,H)\ge \max\{\chi(G),\left\lceil \frac{\chi(G)}{2}\right\rceil +q \}.\]
\end{lemma}
\begin{proof}
Let $H$ be any spanning subgraph of $G$ with $\delta(H)\ge 1$, and let $k=BBC_q(G,H)$. Furthermore, let $c$ be a $q$-backbone $k$-coloring of $G$. Since any $q$-backbone coloring of $(G,H)$ is also a proper coloring of $G$, we have that $k\ge \chi(G)$. Now, if either $q\le\lfloor \chi(G)/2\rfloor$, or $q\ge\lceil \chi(G)/2\rceil$ and $k\ge 2q$, we are done. So, suppose $q\ge \lceil \chi(G)/2\rceil$ and $k<2q$, and let $k'=2q-k$. We claim that $[i]_q=\{1,\ldots,k\}$, for every $i\in\{q-k'+1,\ldots,q\}$. Because $d_H(u)\ge 1$, we know that none of these $k'$ colors can be used on $u$, for every $u\in V(G)$, and the following holds:

\[k-k'=k-2q+k\ge \chi(G). \]
This inequality implies that:

\[k\ge \left\lceil \frac{\chi(G)}{2}\right\rceil+q.\]
It remains to prove our claim. So, let $i$ be any color in $\{q-k'+1,\ldots,q\}$. It suffices to show that $\{1,k\}\subseteq [i]_q$. Clearly, $1\in [i]_q$, since $i\le q$. Also, since $k=2q-k'$ and $i\ge q-k'+1$, we get $k-i \le 2q-k'-q+k'-1 = q-1$. Thus, $k\in[i]_q$ and the lemma follows.
\end{proof}

%\bibliographystyle{elsarticle-num}
%\bibliography{realizes_chi}

\begin{thebibliography}{99}
\bibitem{BM08}
A.~Bondy, U.~Murty, Graph Theory, Graduate Texts in Mathematics,
  Springer-Verlag London, 2008.

\bibitem{BFGW03}
H.~Broersma, F.~V. Fomin, P.~A. Golovach, G.~J. Woeginger, Backbone colorings
  for networks, in: H.~L. Bodlaender (Ed.), Graph-Theoretic Concepts in
  Computer Science, Vol. 2880 of Lecture Notes in Computer Science, Springer
  Berlin Heidelberg, 2003, pp. 131--142.

\bibitem{BFGW07}
H.~Broersma, F.~V. Fomin, P.~A. Golovach, G.~J. Woeginger, Backbone colorings
  for graphs: Tree and path backbones, Journal of Graph Theory 55~(2) (2007)
  137--152.

\bibitem{BB15}
Y.~Bu, X.~Bao, Backbone coloring of planar graphs for $c_8$-free or $c_9$-free,
  Theoretical Computer Science 580 (2015) 50--58.

\bibitem{BL11}
Y.~Bu, Y.~Li, Backbone coloring of planar graphs without special circles,
  Theoretical Computer Science 412~(46) (2011) 6464--6468.

\bibitem{BZ11}
Y.~Bu, S.~Zhang, Backbone coloring for $c_4$-free planar graphs, Sci. Sin. Math
  41 (2011) 197--206.

\bibitem{CHSS13}
V.~Campos, F.~Havet, R.~Sampaio, and A.~Silva.
 Backbone colouring: Tree backbones with small diameter in planar graphs. Theoretical Computer Science 487 (2013), 50--64.

\bibitem{HKLT14}
F.~Havet, A.D.~King, M.~Liedloff, I.~and Todinca. (Circular) backbone colouring: Forest
backbones in planar graphs. Discrete Applied Mathematics 169 (2014), 119--134.

\bibitem{JT95}
T.~R. Jensen, B.~Toft, Graph Coloring Problems, Wiley-Interscience, New York,
  1995.

\bibitem{Karp72}
R.~Karp, Reducibility among combinatorial problems, in: R.~Miller, J.~Thatcher
  (Eds.), Complexity of Computer Computations, Plenum Press, 1972, pp. 85--103.

\bibitem{MR01}
M.~Molloy, B.~Reed, {Graph Colouring and the Probabilistic Method}, 1st
 Edition, Springer, 2001.

\bibitem{WBMR12}
W.~Wang, Y.~Bu, M.~Montassier, A.~Raspaud, On backbone coloring of graphs,
  Journal of combinatorial optimization 23~(1) (2012) 79--93.


\bibitem{BZ10}
S.-M. Zhang, Y.-H. Bu, Backbone coloring for $c_5$-free planar graphs, J. Math
  Study 43~(4) (2010) 315--321.

\end{thebibliography}

\end{document}